\documentclass[11pt]{article}
\usepackage{amsmath,amsfonts,amsthm,amssymb}
\usepackage{url}
\usepackage{graphicx}
\usepackage{lmodern}

\title{A tighter bound on the number of relevant variables in a bounded degree Boolean function}

\author{Jake Wellens}

\usepackage{color}

\usepackage[colorlinks=true, linkcolor=red, urlcolor=blue, citecolor=blue]{hyperref}

\graphicspath{ {images/} }

\newcommand\be{\begin{equation}}
\newcommand\ee{\end{equation}}
\newcommand\bea{\begin{eqnarray}}
\newcommand\eea{\end{eqnarray}}

\newcommand\R{\mathbb{R}}

\newcommand\E{\mathbb{E}}

\newcommand\nn{\nonumber}

\newcommand{\poly}{\text{poly}}

\newcommand{\Inf}{\text{Inf}}

\newcommand{\bls}{\text{bs}}

\newtheorem{thm}{Theorem}

\newtheorem{lem}[thm]{Lemma}
\newtheorem{prop}[thm]{Proposition}
\newtheorem{cor}[thm]{Corollary}
\newtheorem{fact}[thm]{Fact}

\newcommand{\bs}[1]{\boldsymbol{#1}}

\advance \topmargin by -\headheight
\advance \topmargin by -\headsep
\evensidemargin \oddsidemargin
\begin{document}
\maketitle

\begin{abstract}
A classical theorem of Nisan and Szegedy says that a boolean function with degree $d$ as a real polynomial depends on at most $d2^{d-1}$ of its variables. In recent work by Chiarelli, Hatami and Saks, this upper bound was improved to $C \cdot 2^d$, where $C = 6.614$. Here we refine their argument to show that one may take $C = 4.416$. 
\end{abstract}

\section{Introduction}

Given a Boolean function $f: \{0,1\}^n \to \{0,1\}$, there is a unique multilinear polynomial in $\R[x_1, \dots, x_n]$ which agrees with $f$ on every input in $\{0,1\}^n$. One important feature of this polynomial is its degree, denoted $\deg(f)$, which is known to be polynomially related to many other complexity measures, such as block sensitivity $\bls(f)$, certificate complexity $C(f)$, decision tree depth $D(f)$, and approximate degree $\widetilde{\deg(f)}$ (see \cite{NS} and \cite{BdW}). One can also bound the number of relevant variables of $f$ (i.e. the variables which actually show up in a term with non-zero coefficient in the polynomial for $f$, also called the \emph{junta size} of $f$) entirely in terms of the degree: 

\begin{thm}[Nisan-Szegedy \cite{NS}]\label{NS} A function $f: \{0,1\}^n \to \{0,1\}$ with degree $d$ has at most $\frac{d}{2}\cdot 2^d$ relevant variables.
\end{thm}

The idea of Nisan and Szegedy's original proof is to lower bound the \emph{influence} of a relevant variable: a polynomial of degree $d$ has total influence at most $d$, and yet the derivative in the direction of a relevant coordinate is a degree $d-1$ polynomial which is not identically zero, so it is non-zero on a random input with probability at least $1/2^{d-1}$. In other words, each relevant coordinate has influence at least $1/2^{d-1}$, so there can be at most $d\cdot 2^{d-1}$ of them.

Theorem \ref{NS} is has the correct exponential dependence on $d$ -- indeed, consider the function $f$ given by the complete binary decision tree of depth $d$ which queries a distinct coordinate at each vertex. This function has degree $d$ and $2^d - 1$ relevant variables. However, it remained open whether the multiplicative factor of $\Theta(d)$ was necessary until a recent paper by Chiarelli, Hatami and Saks showed that $O(2^d)$ suffices.\footnote{In the same paper, the authors also give an improved lower bound construction, namely, for each $d$, a degree $d$ function with $\frac{3}{2}\cdot 2^d - 2$ relevant variables. }

\begin{thm}[Chiarelli, Hatami, Saks, \cite{CHS}]\label{6.6}
A function $f: \{0,1\}^n \to \{0,1\}$ with degree $d$ has at most $(6.614)\cdot 2^d$ relevant variables.
\end{thm}

The main idea in \cite{CHS} is to replace influence by a different measure -- one which behaves more stably under restrictions of variables. Specifically, they define

$$W(f) = \sum_{i \in R(f)} 2^{-\deg_i(f)}$$
where $R(f)$ is the set of relevant variables for $f$, and $\deg_i(f)$ is the degree of $f$ in $x_i$. It is straightforward to check that $W(f)$ does not decrease by more than $|H|\cdot 2^{-d}$ in expectation when randomly restricting a set $H$ of coordinates with $\deg_i(f) = \deg(f) = d$. If $H$ is chosen well, these contributions are summable, and hence $W(f)$ is bounded above by some universal constant. Since $|R(f)| \leq 2^{\deg(f)}W(f)$, this implies Theorem \ref{6.6}.

The heart of the proof is therefore in choosing the set $H$ which is both \emph{small} enough so that $W(f)$ does not incur a heavy loss, and yet \emph{significant} enough that the restricted functions are of reduced complexity. The idea used in \cite{CHS} (originating in unpublished work of Nisan and Smolensky, see \cite{BdW}) is to build $H$ from a maximal collection of disjoint monomials of full degree. The number of such disjoint monomials is limited by the block sensitivity $\bls(f)$, which is always at most $d^2$, and by maximality, all of the resulting restricted functions have degree $\leq d - 1$.

\subsection{Our improvements}

The above idea certainly does the trick, but there are two somewhat substantial sources of slack in the analysis: one is the global use of the worst-case bound $\bls(f) \leq d^2$, which can be improved for any fixed $d$ with a finite computation. The other is that restricting a large disjoint collection of degree $d$ monomials actually causes a large drop in block sensitivity, which can be exploited. By leveraging both of these ideas, we are able to improve the constant $6.614$ by about 33$\%$:

\begin{thm}[Main result]\label{main}
A function $f: \{0,1\}^n \to \{0,1\}$ with degree $d$ has at most $(4.416)\cdot 2^d$ relevant variables.
\end{thm}

\section{Preliminaries}

\textbf{Restrictions:} For a function $f:\{0,1\}^n \to \{0,1\}$, a set $H \subset[n]$, and an assignment $\alpha: H \to \{0,1\}$, we denote by $f_{\alpha}$ the restricted function obtained by setting the variables $x_h$ to $\alpha(h)$ for $h \in H$. We will sometimes use $f(\alpha_H, x)$ for $f_\alpha(x)$ if we want to be explicit about the set of coordinates which have received the assignment by $\alpha$. 

\textbf{Influence:} The influence of coordinate $i$ on $f$, or $\Inf_i[f]$, is the probability that, for a uniformly random input $x$, flipping the $i$th bit of $x$ causes the value of $f(x)$ to flip. The total influence $\Inf[f] = \sum_{i \in [n]} \Inf_i[f]$ can also be expressed in terms of the Fourier coefficients of $f$, namely
$$ \Inf[f] = \sum_{S \subseteq [n]} |S|\widehat{f}(S)^2.$$ 
Since the degree of $f$ remains unchanged when $f$ is expressed as a multilinear polynomial over $\{0,1\}^n$ (as we consider in this paper) or $\{1, -1\}^n$ (as in the Fourier expansion), the above formula makes it clear that a Boolean function of degree $d$ has $\Inf[f] \leq d$. As mentioned in the introduction, the following useful fact is from \cite{NS}, and can be proved by induction:
\be\label{NS_fact} \Inf_i[f] \geq 2^{1-\deg_i(f)} . \ee

\textbf{Block sensitivity:} For a set $B \subset [n]$ and a string $x \in \{0,1\}^n$, we denote by $x^B$ the string obtained from $x$ by flipping all the bits $x_b$ for $b \in B$. Recall that the block sensitivity of $f$ at an input $x$ (denoted $\bls_x(f)$) is the maximum number $b$ of disjoint blocks $B_1, \dots, B_b \subset [n]$ such that $f(x) \neq f(x^{B_i})$ for all $1 \leq i \leq b$, and the block sensitivity of $f$ (denoted $\bls(f)$) is the maximum of $\bls_x(f)$ over all inputs $x$.  It is well-known that block sensitivity and degree are polynomially related:
\be\label{bs_vs_deg} \deg(f)^{1/3} \leq \bls(f) \leq \deg(f)^2 \ee
although neither bound is known to be sharp. The best known constructions have $\bls(f) = \Theta(\deg(f)^{1/2})$ and $\bls(f) = \Theta(\deg(f)^{\log_3(6)}) = \Theta(\deg(f)^{1.6309...}) $ respectively. See \cite{BdW} and  \cite{HKP} for details and for relationships to many other complexity measures.

\textbf{The measure $W(f)$:} Recall that $$W(f) := \sum_{i \in R(f)} 2^{-\deg_i(f)},$$ where $R(f)$ is the set of relevant coordinates (i.e. coordinates $i$ for which $\Inf_i[f] > 0$) and $\deg_i(f)$ is the degree of largest degree monomial appearing  in $f$ (with non-zero coefficient) that contains $x_i$.  The behavior of $W$ under restrictions boils down to the following inequality, whose simple proof we reproduce below for completeness.

\begin{fact}[\cite{CHS}]\label{fact} For any relevant coordinates $i \neq j$, let $f_0$ and $f_1$ be the restrictions obtained from $f$ by setting $x_j$ to 0 and 1 respectively. Then
\be \label{key_ineq}2^{-\deg_i(f)} \leq 2^{-\deg_i(f_0) - 1} + 2^{-\deg_i(f_1) - 1} \ee
\end{fact}
\begin{proof}
Write $f = x_jf_1 + (1-x_j)f_0 = x_j(f_1 - f_0) + f_0$, from which it is clear that $\deg_i(f) \geq \deg_i(f_0)$. If $\deg_i(f) \geq 1 + \deg_i(f_0)$ then the inequality is true independently of $\deg_i(f_1)$. Otherwise, it must be that $\deg_i(f) = \deg_i(f_0)$, in which case the leading degree monomials for $x_i$ must cancel in $f_1 - f_0$. But this implies $\deg_i(f) = \deg_i(f_1) = \deg_i(f_0)$, and so the inequality becomes an equality in this case.
\end{proof}
By summing (\ref{key_ineq}) over $i \in R(f)$ and iterating over restrictions of more variables, one obtains

\be \label{summed_ineq} W(f) \leq |H|\cdot 2^{-d} + \frac{1}{2^{|H|}}\sum_{\alpha: H \to \{0,1\}} W(f_\alpha) \ee
for any set $H \subset [n]$ with $\deg_i(f) = d$ for all $i \in H$. As in \cite{CHS}, we define $$W_d := \max_{\deg(f) = d} W(f).$$ If $H$ is chosen as a maximal collection of degree $d$ monomials in $f$, then each $f_\alpha$ has degree at most $d-1$. An unpublished argument of Nisan and Smolensky (which we essentially use in the proof of Lemma \ref{bd_recursive} below) implies that $|H| \leq \deg(f)\cdot \bls(f) \leq d^3$, and so  (\ref{summed_ineq}) yields the recursive inequality
$$W_d \leq d^3 \cdot 2^{-d} + W_{d-1}.$$
This is already summable, but the bound $W(f) \leq \frac{1}{2}\Inf[f] \leq d/2$ is preferable for small $d$, and optimizing over the choice of the two bounds yields $W_d \leq 6.614$ for all $d$.

\section{Improving the constant}

\subsection{Don't spend it all in one place}

Our first new idea is simply to keep track of block sensitivity through the restriction process: the main observation is Proposition \ref{bs decr} below, which says that if $f$ has $\ell$ disjoint monomials of maximum degree, then by assigning any values to the variables in these monomials, the block sensitivity of the restricted function decreases by $\ell$. So, if we have to restrict many variables in order to drop the degree of $f$ (i.e. to hit all the maximum degree monomials), then we must ``spend" our limited supply of block sensitivity, and in the future it will become much easier to lower the degree again. \\

\begin{prop}\label{bs decr}
If $f: \{0,1\}^n \to \{0,1\}$ has $\ell$ \textbf{disjoint} monomials $M_1, \dots, M_\ell$, each of  degree $d = \deg(f)$, then for any assignment $\alpha: \cup M_i \to \{0,1\}$, the restricted function $f_\alpha$ has \emph{$$\bls(f_\alpha) \leq \bls(f) - \ell.$$ }
In particular, \emph{$\ell \leq \bls(f)$}.
\end{prop}

\begin{proof}
Let $M = \cup_{i=1}^\ell M_i$ and $b = \bls(f_{\alpha}) = \bls_y(f_\alpha)$, for some $y \in \{0,1\}^{[n] \setminus M}$. Then there are $b$ disjoint blocks $B_1, \dots, B_b \subset [n] \setminus M$ with $f(\alpha_M, y) \neq f(\alpha_M, y^{B_j})$ for each $j$. Since $M_i$ is a maximum degree monomial in $f$, each of the functions $\{0,1\}^{M_i} \ni x  \mapsto f(x, z)$ is non-constant for any $z$. Therefore, for each $i$, there is a block $C_i \subset M_i$ with $f(\alpha_M^{C_i}, y) \neq f(\alpha_M, y)$. Therefore $\{C_1, \dots, C_\ell, B_1, \dots, B_b\}$ is a collection of disjoint sensitive blocks for $f$ at the input $(\alpha_M, y) \in \{0,1\}^n$, and so $\bls(f) \geq b + \ell$.
\end{proof}

To keep track of $W$, degree and block sensitivity simultaneously, we define
 $$W(b, d) := \max_{\substack{f \text{ with }\bls(f) \leq b  \\ \text{ and } \deg(f) = d }} W(f).$$ 
 
Note that $W(0, d) = 0$, $W(b, 0) = 0$, and $W(b, d) \leq W_d$ for any $b$. By (\ref{bs_vs_deg}), we have $W(d^2, d) = W_d$, and we make the convention that $W(b, d) = 0$ for $b > d^2$. 

\begin{lem}\label{bd_recursive}
For each $b,d$ with $b \leq d^2$, we have
$$W(b,d) \leq  \max_{(\ell, k) \in \{1, \dots, b\} \times \{1, \dots, d \}}\left(\ell \cdot d  \cdot2^{-d} + W(b - \ell, d - k)\right) $$
\end{lem}

\begin{proof}
Suppose $f$ has degree $d$ and $\bls(f) \leq b$. Let $M_1, \dots, M_\ell$ be a maximal collection of disjoint degree $d$ monomials in $f$, and let $H = \cup_i M_i$. By inequality (\ref{summed_ineq}), 

$$W(f) \leq \underbrace{|H|\cdot 2^{-d}}_{= \, \ell \cdot d \cdot 2^{-d}} + \underset{\alpha: H \to \{0,1\}}{\E}[W(f_\alpha)]$$

Because the collection $\{M_1, \dots, M_\ell\}$ is maximal, $H$ hits every degree $d$ monomial and hence each $f_\alpha$ has degree $d_{\alpha} \leq d - 1$. By Proposition \ref{bs decr}, each $f_{\alpha}$ has $\bls(f_\alpha) \leq b - \ell$. Since $W(\cdot ,d)$ is monotone (for feasible inputs), it follows that for each $\alpha$, $W(f_\alpha) \leq W(b - \ell, d - k')$, where $k' = \arg\max_{k \in \{1, \dots, d\}} W(b - \ell, d - k)$. Taking the maximum over all possible values of $\ell \in \{1, \dots, b\}$ yields the desired bound.
\end{proof}

Since $W_d$ is bounded and increasing\footnote{It is shown in \cite{CHS} that $W_d \geq 2^{-d} + W_{d-1}$, and in fact their lower bound construction can be turned into a proof that $W_d \geq 2\cdot 2^{-d} + W_{d-1}$.}, so $W^* := \lim_{d \to \infty} W_d$ exists. Since $W_d = W(d^2, d)$, the following corollary comes easily from Lemma \ref{bd_recursive}.  
\begin{cor}\label{bd_cor}
For any $d$, 
$$W^* \leq W(d^2, d) + \sum_{r = d+ 1}^{\infty} r^3 2^{-r} $$
\end{cor}

Lemma \ref{bd_recursive} yields explicit bounds on $W(b,d)$ for any finite $(b,d)$, which in turn yields an explicit bound on $W^*$ via Corollary \ref{bd_cor}. For small values ($d \leq 9$), the bound 
$$W(b,d) \leq W_d \leq \max _{\deg(f) = d}\sum_{i \in R(f)} \frac{\Inf_i[f]}{2} = \frac{d}{2}  $$
is better than the one from Lemma \ref{bd_recursive}. Extracting numerical bounds recursively yields
$$W(50^2, 50) \leq 5.07812...$$
which implies the same bound (to around 10 decimal digits) on $W^*$.

\subsection{Tighter bounds on block sensitivity for low degree functions}

To further reduce our estimate of $W^*$, we focus on functions of low degree, which clearly have the most influence on the bounds. Specifically, we produce sharper upper bounds on the block sensitivity of such functions, by solving a small set of linear programs. We begin with a simple reduction to linear program feasibility, using ideas from the original proof of $\bls(f) \leq 2\deg(f)$ from \cite{NS}.\footnote{The ``2" in this bound can be removed by using repeated function composition (or \emph{tensorization}), as shown in \cite{Tal}.}

\begin{fact}\label{bs_reduction}
If there exists a function $f : \{0,1\}^n \to \{0,1\}$ of degree $d$ with block sensitivity $b$, then there exists another function $g: \{0,1\}^b \to \{0,1\}$ of degree $\leq d$ with $g(0) = 0$ and $g(w) = 1$ for each vector $w$ of hamming weight 1.  
\end{fact}

\begin{proof}
If $f(x)$ attains maximal block sensitivity at $z$, then $f(x \oplus z)$ attains maximal block sensitivity at 0, so without loss of generality we may assume $z = 0$, and possibly replacing $f$ by $1-f$ we may also assume that $f(0) = 0$. If $B_1, \dots, B_b$ are sensitive blocks for $f$ at 0, then define $$g(y_1, \dots, y_b) = f(\underbrace{y_1, \dots, y_1}_{B_1}, \dots, \underbrace{y_b, \dots, y_b}_{B_b})$$
so that for each coordinate vector $e_i$, $g(e_i) = f(\textbf{1}_{B_i}) = f(0^{B_i}) = 1$.
\end{proof}

For any $d \geq 1$, define the moment map $m_d: \R \to \R^d$ by $m(t) = (t, t^2,\dots, t^d)$.

\begin{prop}
If there exists a degree $d$ function $f: \{0,1\}^n \to \{0,1\}$ with block sensitivity $b$, then there exists $\tau \in \{0,1\}$ such that the following set of linear inequalities has a solution $p \in \R^d$:
\bea \nn
\langle p, m_d(1) \rangle &=& 1 \\ \label{LP}
0 \leq \langle p, m_d(k) \rangle  &\leq & 1 \,\, \text{ for each } k \in \{2, \dots, b-1\} \\ \nn
\langle p, m_d(b) \rangle &=& \tau
\eea
\end{prop}

\begin{proof}
If such an $f$ exists, then let $q(x_1, \dots, x_b) = \frac{1}{b!}\sum_{\sigma \in S_b}g(x_{\sigma(1)}, \dots, x_{\sigma(b)})$, where $g$ comes from Fact \ref{bs_reduction}, and set $\tau = g(1, 1, \dots, 1)$. It is well known (see \cite{BdW}) that there is a univariate polynomial $p: \R \to \R$ of degree at most $d$ such that for any $x \in \{0,1\}^b$, $q(x_1, \dots, x_b) = p(x_1 + \dots + x_b)$. For each $k \in \{1, \dots, b\}$, $p(k)$ is therefore the average value of $g$ on boolean vectors with hamming weight $k$, so in particular $p(k) \in [0,1]$. We also know $p(0) = g(0) = 0$, $p(b) = g(1,\dots, 1) = \tau$, and $p(1) = \frac{1}{n}\sum_i g(e_i) = 1$, and hence the coefficients of $p$ provide a solution to the set of linear inequalities.  
\end{proof}

Using the simplex method with exact (rational) arithmetic in Maple, we compute the largest $b=b(d)$ for which the LP (\ref{LP}) is feasible for $1 \leq d \leq 14$, which yields upper bounds on block sensitivity for low degree boolean functions. These bounds are summarized in Table \ref{bs table}.

\begin{table}[]
\resizebox{0.9\textwidth}{!}{%
\begin{tabular}{|l|l|l|l|l|l|l|l|l|l|l|l|l|l|l|}
\hline
$\deg(f)$ &1 & 2 & 3 & 4 & 5 & 6 & 7 & 8 & 9 & 10 & 11 & 12 & 13 & 14 \\ \hline
$\bls(f) \leq$ &1 & 3 & 6 & 10 & 15 & 21 & 29 & 38 & 47 & 58 & 71 & 84 & 99 & 114 \\ \hline

\end{tabular}%
}
\caption{LP bounds on block sensitivity for low degree functions.}
\label{bs table}
\end{table}

Setting $W(b, d) = 0$ for $b > b(d)$, for each $d = 1, \dots, 14$, we can recompute the recursive bounds from Lemma \ref{bd_recursive} and obtain
$$W(30^2, 30) \leq 4.41571 \implies W^* \leq 4.4158.$$

\noindent\textbf{Remark:} The largest known separation between block sensitivity and degree is exhibited by a function on 6 variables with degree 3. By a tensorization lemma in \cite{Tal}, any degree $d$ function $f$ with block sensitivity $b$ yields an infinite family of boolean functions $f_k$ with $\deg(f_k) = d^k$ and $\bls(f_k) \geq b^k$. Hence, if an entry $(d, b(d))$ in Table \ref{bs table} is tight for some $d \geq 4$, then by Fact \ref{bs_reduction} there is a function on $b(d)$ variables exhibiting a larger-than-currently-known separation between degree and block sensitivity. If $\bs(f) = \deg(f)^{\log_3(6)}$ is in fact the optimal separation, then our techniques would show $W^* < 3.96$.

\begin{table}[]
\resizebox{\textwidth}{!}{%
\begin{tabular}{|l|l|l|l|l|l|l|l|l|l|l|l|l|l|l|}
\hline
$\deg(f)$ &1 & 2 & 3 & 4 & 5 & 6 & 7 & 8 & 9 & 10 & 11 & 12 & 13 & 14 \\ \hline
$W(f) \leq$ &0.5 & 1 & 1.5 & 2.0 & 2.5 & 3.0 & 3.5 & 3.9375 & 4.096 & 4.203 & 4.273 & 4.311 & 4.335 & 4.348 \\ \hline

\end{tabular}%
}
\caption{Bounds on $W(f)$ for low degrees, obtained using Lemma 6 and Table \ref{bs table}. }
\label{W table}
\end{table}

\section{Discussion and concluding remarks}\label{other}

While our methods are unlikely to produce the optimal $W^*$, they do suggest a few interesting questions.
\begin{itemize}

\item The proof seems to suggest that block sensitivity limits junta size, and for small $d$, the values of $W(b,d)$ are much lower when $b \ll d^2$ than when $b \sim d^2$. A classical result of Simon \cite{Simon} says that $|R(f)| \leq s(f)4^{s(f)}$. Interestingly, we can obtain a proof of a weaker version of Simon's theorem using only the techniques in this paper and in \cite{CHS}. The idea is to define an analogue of $W$ for sensitivity instead of degree:
$$S(f) := \sum_{i \in R(f)} 2^{-s_i(f)}, \, \, \, \text{ where } s_i(f) := \max_{\{x \,: \,f(x^i) \neq f(x)\}} (s_x(f)+s_{x^i}(f)). $$
\textbf{Claim:} \textit{Fact \ref{fact} holds with $\deg_i(f)$ replaced by $s_i(f)$.} 

\textit{Proof:} Since $s_i(f) \geq \max\{s_i(f_0), s_i(f_1)\}$ is clear from the definitions, we can assume that $f_0$ does not depend on $x_i$.  Without loss of generality suppose $j = 1$ and that $y$ has $f(1, y) = 1 \neq f(1, y^i)$ and $s_i(f_1) = s_y(f_1) + s_{y^i}(f_1)$. First suppose $f_0(y) = 0$. Since $f(1, y) = 1$, this means $f$ is also sensitive to $j$ at input $(1, y)$, and so $s_i(f) \geq 1 + s_i(f_1)$. If $f_0(y) = 1$, then $f_0(y^i)$ is also 1 because $f_0$ does not depend on $x_i$. But then $f$ is sensitive to $j$ at input $(1,y^i)$, and so either way $s_i(f) \geq 1 + s_i(f_1)$, which implies the claim. \qed

From here we arrive at the analogue of (\ref{summed_ineq}), and we can proceed in a number of ways. As in the proof of Lemma \ref{bd_recursive}, we can restrict maximum degree monomials until we run out of block sensitivity, yielding $|R(f)| \leq \deg(f)\cdot\bls(f)\cdot 4^{s(f)}$. 

In any case, it seems reasonable to conjecture a Nisan-Szegedy theorem for block sensitivity, namely that any boolean function $f$ is a $\poly(\bls(f))\cdot2^{\bls(f)}$-junta.

 \item More generally, it would be interesting to characterize certain \emph{ternary} relationships between complexity measures. Many of the examples we know which achieve optimal or best-known separations between two measures tend to have the property that a third measure is equal or very close to one of the other two. (For example, the best known gap of the form $\bls(f) \ll \deg(f)$ is attained by a Tribes function on $n$ variables with $\deg(f) = n$ and $\bls(f) = s(f) = C(f) = \sqrt{n}$.) Moreover, these examples almost always have $|R(f)| = \poly(\deg(f))$. Meanwhile, the known examples of functions with nearly-optimal junta size do not exhibit any super-constant separation between the measures $\deg(f), \bls(f), s(f)$ and $D(f)$. (It is possible to hybridize small, well-separated functions with large juntas, but the separations and the junta size both suffer some loss.)

 \item Finally, Table \ref{bs table} suggsts that $\bls(f) \leq c_0\deg(f)^2$, for $c_0 \approx 0.59$. If you enjoyed reading this paper, perhaps you would enjoy trying to compute the optimal value of $c_0$. One consequence of showing that $\bls < \deg^2$ is that any separation between $\bls$ and degree proven by simply tensorizing a single example would necessarily (in the absence of more sophisticated arguments) look like $\bls(f) \geq \deg(f)^{2 - \epsilon}$, for some $\epsilon > 0$.
\end{itemize}

\end{document}